\documentclass[envcountsect,envcountsame,11pt]{llncs}

\newif\ifcameraready
\camerareadyfalse
\ifcameraready

\else
  \usepackage[margin=1in]{geometry}                 		
\fi
\usepackage{graphicx}				
\usepackage{amsmath,amssymb}
\let\doendproof\endproof
\renewcommand\endproof{~\hfill\qed\doendproof}

\usepackage{bbm}
\usepackage{braket}
\usepackage{mathtools}
\usepackage{dsfont}
\usepackage{mathrsfs}
\usepackage{hyperref}
\usepackage{color}
\usepackage[capitalise]{cleveref}

\title{Efficient Simulation for \\Quantum Message Authentication}
\author{Anne Broadbent \and Evelyn Wainewright}

\institute{Department of Mathematics and Statistics, University of Ottawa, Canada,\\
\email{\{abroadbe,ewain031\}@uottawa.ca}
}

\newtheorem{defn}{Definition}

\newcommand{\C}{\mathcal{C}}
\newcommand{\D}{\mathcal{D}}
\newcommand{\E}{\mathcal{E}}

\newcommand{\K}{\mathcal{K}}

\newcommand{\revise}[1]{}

\newcommand{\reg}[1]{\ensuremath{#1}}
\newcommand{\regS}{\reg{S}}
\newcommand{\regC}{\reg{C}}
\newcommand{\regM}{\reg{M}}

\newcommand{\regF}{\reg{F}}
\newcommand{\regR}{\reg{R}}

\newcommand{\trnorm}[1]{\left\|#1\right\|_1}
\newcommand{\abs}[1]{\left|#1\right|}

\newcommand{\ketbra}[2]{\ket{#1}\bra{#2}}			
\newcommand{\egoketbra}[1]{\ketbra{#1}{#1}}			

\newcommand{\accstate}{\egoketbra{\text{acc}}}
\newcommand{\rejstate}{\egoketbra{\text{rej}}}
\newcommand{\accstatepure}{\ket{\text{acc}}}
\newcommand{\rejstatepure}{\ket{\text{rej}}}

\begin{document}
\maketitle

\begin{abstract}
Quantum message authentication codes are families of keyed encoding and decoding maps that enable the detection of
tampering on
encoded quantum data.
Here, we study a new class of simulators for quantum message authentication schemes,  and show how they are applied in the
 context of two codes: the \emph{Clifford} and the \emph{trap} code. Our results show for the first time that these codes admit an \emph{efficient simulation} (assuming that the adversary is efficient). Such efficient simulation is typically crucial in order to establish a composable notion of security.
\end{abstract}

\section{Introduction}
\label{sec:intro}
Quantum cryptography is the study of the security of information processing in a quantum world. While quantum key distribution~\cite{BB84} is today the most widely successful quantum cryptographic technology~\cite{BEM+07,Feh10}, quantum information effectively re-defines many cryptographic paradigms~\cite{BS16}. Among these is the need for new definitions and protocols for cryptographic tasks that operate on quantum data, such as quantum secret sharing~\cite{CGL99} and quantum multi-party computation~\cite{BCG+06}. Another fundamental task is quantum message authentication.

Quantum message authentication schemes, introduced in~\cite{BCG+02}, are families of keyed encoding and decoding maps which allow for the detection of tampering on encoded quantum data. These codes were originally given in a very efficient form, based on \emph{purity testing}~\cite{BCG+02}, and were shown to also satisfy a composable security notion~\cite{HLM11}.

Further quantum message authentication schemes have been proposed, including the \emph{signed polynomial code}\cite{BCG+06,ABE10}, the trap code~\cite{BGS13} and the \emph{Clifford} code~\cite{ABE10,DNS12}. These schemes have a nice algebraic form, which makes them particularly easy to study. Perhaps the main reason for interest in these schemes is that they have a sufficient amount of ``structure'' to enable evaluation of quantum gates over the encoded data (this technique is called \emph{quantum computing on authenticated data (QCAD)}). This has lead to protocols for multi-party quantum computation\cite{BCG+06}, quantum one-time programs \cite{BGS13} and the verification of quantum computations~\cite{ABE10}.

The security of quantum message authentication schemes is typically defined in terms of the existence of a  \emph{simulator} that, given access only to the ideal functionality for quantum message authentication (which is a virtual device that either transmits the quantum data directly and outputs ``accept'', or replaces it with a dummy state and outputs ``reject''), is able to emulate the behaviour of the adversary so that the real-world protocol (involving the adversary) is statistically indistinguishable from the ideal-world protocol (involving the simulator). This type of definition fits in the quantum Universal Composability (UC)\cite{Can01,Unr10}  framework, as long as we add a further condition: if the adversary runs in polynomial time, so must the simulator (an \emph{efficient} simulation). Until now, direct efficient simulations were known only for the purity-testing based codes~\cite{BCG+02}.

In this work, we show a new family of efficient simulators for quantum message authentication schemes. The main idea is that the simulator replaces the entire codeword by half-EPR pairs (keeping the remaining half to itself), and runs the adversary on these entangled states (as well as the reference system for the original input). After the attack is applied, the simulator performs Bell basis measurements in order to verify the integrity of the EPR pairs. So long as enough EPR pairs are found to be intact, the simulator makes the ideal functionality ``accept''; otherwise, it makes it ``reject''.  It is well-known that this Bell basis measurement will detect any non-identity Pauli attack---given the structure of the codes that we analyze, we show that this is sufficient.

We apply this type of simulator to the Clifford and trap quantum message authentication codes.  We note that the Clifford code was previously proven secure according to an algebraic definition, without an efficient simulation~\cite{ABE10,DNS12}, and that the trap scheme was proven secure according to a simulator for a more elaborate ideal functionality for \emph{quantum one-time programs}~\cite{BGS13}. We thus establish for the first time efficient simulators for these codes (note, however, that we make extensive use of the algebraic tools developed in these prior works, and that we achieve the same security bounds).
We also note that the idea of using EPR-pair testing as a proof technique for quantum message authentication has appeared in~\cite{BCG+02}, where a more elaborate type of testing (called \emph{purity testing}) is used.

\paragraph{\textbf{Roadmap.}} The remainder of the paper is structured as follows. In Section~\ref{sec:prelims}, we give some details on the standard notation and well-known facts that are used throughout. In Section~\ref{sec:defns}, we formally define quantum message authentication in terms of correctness and security. Section~\ref{sec:schemes} gives the Clifford and trap schemes, while in Section~\ref{sec:security} we show security of the schemes.

\section{Preliminaries}
\label{sec:prelims}
Here, we present  basic notation (\cref{sec:basic}) and well-known facts about the Pauli (\cref{sec:Paulis})  and Clifford (\cref{sec:Cliffords})~groups.

\subsection{Basic Notation}
\label{sec:basic}
We assume the reader is familiar with the basics of quantum information~\cite{NC00}, but nevertheless give a quick review of the most relevant notation in this section. We will use the density operator formalism to represent quantum states. Density matrices are represented with a greek letter, typically~$\rho$. The subscripts of the quantum states indicate which spaces (registers) the states reside in. We therefore represent the density operator for the state in the $M$ register as~$\rho_M$.

The \emph{trace norm} of a state, $\rho$, denoted $\trnorm{\rho}$, is defined as $\trnorm{\rho}=tr[\sqrt{\rho^{\dagger}\rho}]$. The \emph{trace distance} between two states $\rho$ and $\sigma$, denoted $D(\rho, \sigma)$, is defined as $D(\rho,\sigma)=\frac{1}{2}\trnorm{\rho-\sigma}$. The trace distance is a measure of distiguishability between the two states $\rho$ and $\sigma$. The trace distance is equal to $0$ if and only if $\rho$ and $\sigma$ are the same state (and therefore indistinguishable) and the trace distance is equal to $1$ if and only if $\rho$ and $\sigma$ are orthogonal (and therefore perfectly distinguishable). The trace norm, and therefore the trace distance, satisfies the triangle inequality: $\trnorm{\rho+\sigma} \leq \trnorm{\rho} + \trnorm{\sigma}$.

Let $\mathcal{B}(\mathcal{H})$ be the space of bounded linear operators acting on a Hilbert space, $\mathcal{H}$. Given $\mathcal{A} \subseteq \mathcal{B}(\mathcal{H}_1)$ and $\mathcal{B} \subseteq \mathcal{B}(\mathcal{H}_2)$ then given a linear map $T$ from $\mathcal{A} \rightarrow \mathcal{B}$, $T$ is called \emph{positive} if $T(A) \geq 0$ for all positive $A \in \mathcal{A}$.  $T$ is a \emph{completely positive} map, (\emph{CP} map), if $T \otimes Id: \mathcal{A} \otimes \mathcal{B} \rightarrow \mathcal{B}(\mathcal{H}_1) \otimes \mathcal{B}(\mathds{C}^n)$ is positive for all $n \in \mathds{N}$. In this case, $Id$ is the identity map on $\mathcal{B}(\mathds{C}^n)$ and $\mathds{C}^n$ is isomorphic to a complex Hilbert space of dimension $n$. A map, $T$, is \emph{trace preserving} if  $tr(T(\rho))=tr(\rho)$. $T$ is a \emph{quantum channel} if it is a \emph{completely positive} and \emph {trace preserving} map (\emph{CPTP} map). A family of quantum maps is \emph{polynomial-time} if they can be written as a polynomial-time uniform family of quantum circuits. A quantum state is \emph{polynomial-time generated} if it given as the output of a polynomial-time quantum map (which takes as input the all-zeros state)~\cite{Wat11}.

A permutation map, denoted throughout by $\pi$, is a unitary operation that acts on $n$ qubits and permutes the order of the $n$ qubits. This can equivalently be seen as a permutation, $\sigma$, of the indices of the qubits, where~$\pi$ would take the $i^{th}$ qubit to the $\sigma(i)^{th}$ position. Permutation maps are orthogonal, real valued matrices so $\pi^{-1}=\pi^{\dagger}$. We use $\Pi_{n}$ to denote the set of all permutation maps on $n$ qubits.

We denote a two-qubit maximally entangled pure state as $\ket{\Phi^{+}}=\frac{1}{\sqrt{2}}(\ket{00}+\ket{11})$. This is one of four Bell states. The other three Bell states are also maximally entangled pure states, $\ket{\Phi^{-}}=\frac{1}{\sqrt{2}}(\ket{00}-\ket{11})$, $\ket{\Psi^{+}}=\frac{1}{\sqrt{2}}(\ket{01}+\ket{10})$, and $\ket{\Psi^{-}}=\frac{1}{\sqrt{2}}(\ket{01}-\ket{10})$. The four Bell states are orthogonal and form a basis for two-qubit states. The four Bell states are therefore perfectly distinguishable and so we can perform a projective measurement into the Bell basis and determine which of the four Bell states we have. This is called a \emph{Bell basis measurement}.

An $[[n,1,d]]$-code is a quantum error correcting code that encodes one logical qubit into $n$ qubits and has distance $d$; if $d=2t+1$, the code can correct up to $t$  bit or phase flips. We assume that the decoding map can always be applied, but if more than $t$ errors are present, it is not guaranteed to decode to the original input.

\subsection{Pauli Matrices}
\label{sec:Paulis}

The single-qubit \emph{Pauli matrices} are given by:
\begin{equation}
  I=\begin{bmatrix} 1 & 0\\ 0 &1 \end{bmatrix}, X=\begin{bmatrix} 0 & 1 \\ 1 & 0 \end{bmatrix}, Z=\begin{bmatrix} 1 & 0\\ 0 & -1 \end{bmatrix}, \text{ and } Y=iXZ=\begin{bmatrix} 0 & -i \\ i & 0 \end{bmatrix}.
  \end{equation}
Recall that if we allow complex coefficients, the any single-qubit gate can be written as a linear combination of the four single-qubit Pauli matrices.

An $n$-qubit Pauli matrix is given by the $n$-fold tensor product of single-qubit Paulis. We denote the set of all $n$-qubit Pauli matrices by $\mathds{P}_n$, where $\abs{\mathds{P}_n}=4^n$. Any $n$-qubit unitary operator, $U$, can also be written as a linear combination of $n$-qubit Paulis, again allowing for complex coefficients. This gives $U= \sum_{P \in \mathds{P}_n} \alpha_P P$, with $\sum_{P \in \mathds{P}_n}|\alpha_P|^2 =1$, since $U$ is unitary. This is called the \emph{Pauli decomposition} of a unitary quantum operation.

The \emph{Pauli weight} of an $n$-qubit Pauli, denoted $\omega(P)$, is the number of non-identity Paulis in the $n$-fold tensor product. We will also define sets of Paulis composed only of specific Pauli matrices, such as $\{I,X\}^{\otimes n}$ which is the set of all $n$-qubit Paulis composed of only $I$ and $X$ Paulis, or  $\{I,Z\}^{\otimes n}$ which is the set of all $n$-qubit Paulis composed of only $I$ and $Z$ Paulis. Finally, Paulis are self-inverses, so $P=P^{-1}=P^{\dagger}$.

The following lemma, called the \emph{Pauli Twirl}~\cite{DCEL09}, shows how we can greatly simplify expressions that involve the twirling of an operation by the Pauli matrices:
\begin{lemma}[Pauli Twirl]\label{lem:Pauli-twirl}
Let $P, P'$ be Pauli operators. Then for any $\rho$ it holds that:
\begin{equation*}
\frac{1}{\abs{\mathds{P}_{n}}}\sum_{Q \in \mathds{P}_n} Q^{\dagger} P Q \rho Q^{\dagger} P'^{\dagger} Q = \begin{cases} 0, \text{ }P \neq P' \\ P\rho P^{\dagger} , \text{otherwise}\,.\end{cases}
\end{equation*}
\end{lemma}

\subsection{Clifford Group}
\label{sec:Cliffords}

The \emph{Clifford group}, $\mathcal{C}_n$, on $n$ qubits are unitaries that map Pauli matrices to Pauli matrices (up to a phase of $\pm1$ or $\pm i$). Specifically, if $P \in \mathds{P}_n$, then for all $C \in \mathcal{C}_n$, $\alpha CPC^{\dagger} \in \mathds{P}_n$, for some $\alpha \in \{ \pm 1, \pm i \}$.
Not only do Cliffords map Paulis to Paulis, but they do so with a uniform distribution~\cite{ABE10}:

\begin{lemma}[Clifford Randomization]\label{lem:Cliff-decomp}
Let $P$ be a non-identity Pauli operator. Applying a random Clifford operator (by conjugation) maps it to a Pauli operator chosen uniformly over all non-identity Pauli operators. More formally, for every $P$, $Q$ $\in \mathds{P}_n \setminus \{ \mathds{I} \}$, it holds that:
\begin{equation*}
\abs{\{ C \in \mathcal{C}_n | C^{\dagger}PC=Q\}}=\frac{\abs{\mathcal{C}_n}}{\abs{\mathds{P}_n}-1}.
\end{equation*}
\end{lemma}

We also state a lemma that is analogous to the Pauli twirl, the \emph{Clifford Twirl}~\cite{DCEL09}; for completeness, the proof is given in~\cref{sec:app-Clifford}.
\begin{lemma}[Clifford Twirl]\label{lem:Cliff-twirl}
Let $P\neq P'$ be Pauli operators. For any~$\rho$ it holds that:
\begin{equation*}
\sum\limits_{C \in \mathcal{C}_n} C^{\dagger}PC \rho C^{\dagger}P'C =0.
\end{equation*}
\end{lemma}
Finally, we note that sampling a uniformly random Clifford can be done efficiently \cite{Got97}.

\section{Quantum Message Authentication}\label{sec:defns}

Following \cite{DNS12}, we define a \emph{quantum message authentication scheme} as a pair of encoding and decoding maps that satisfy the following:
\begin{defn}[Quantum message authentication scheme] \label{defn:QAS}
A \emph{quantum message authentication scheme} is a polynomial-time set of encoding and decoding channels $\{(\E_k^{\regM \rightarrow \regC},\D_k^{\regC \rightarrow \regM\regF}) \mid k \in \K \}$, where $\K$ is the set of possible keys,~$\regM$ is the input system, $\regC$ is the encoded system, and $\regF$ is a flag system that is spanned by two orthogonal states: $\accstatepure$ and $\rejstatepure$, such that for all $\rho_M$,  $(\D_k \circ \E_k)(\rho_M)=\rho_M\otimes \accstate$.
\end{defn}

In order to define security for a quantum message authentication scheme, we first consider a reference system~$\regR$, so that the input can be described as $\rho_{\regM\regR}$ and we can furthermore assume that the system consisting of the encoded message, together with the reference system, undergoes a unitary adversarial attack $U_{\regC\regR}$.  For a fixed key, $k$, we thus define the \emph{real-world} channel as:
\begin{equation}\label{eqn:real-world}
{\mathscr{E}_k}^{\regM\regR \rightarrow \regM\regR\regF}: \rho_{\regM\regR} \mapsto (\mathcal{D}_{k}\otimes \mathds{I}_R)(U_{\regC\regR} (\mathcal{E}_{k} \otimes \mathds{I}_R)(\rho_{\regM\regR})U_{\regC\regR}^{\dagger}),
\end{equation}
where $\mathds{I}_R$ is the identity map on the reference system, $R$. From now on, we will not include the identity maps, since it will be clear from context which system undergoes a linear map and which one does not.

Security is given in terms of the existence of a \emph{simulator}, which has access only to the ideal functionality. This ideal functionality either accepts (and leaves the message register $\regM$ intact), or rejects (and outputs a fixed state $\Omega_{M}$); the simulator can interact with the ideal functionality by selecting \emph{accept} or \emph{reject}. In both cases, the simulator can also alter the reference system $\regR$. This ideal-world process is modeled by the quantum channel $\mathscr{F}$, called the \emph{ideal} channel, where for each attack, $U_{CR}$, there exists two CP maps $\mathscr{U}^{acc}$ and $\mathscr{U}^{rej}$ acting only on the reference system $\regR$ such that $\mathscr{U}^{acc}+\mathscr{U}^{rej}=\mathds{1}$:
\begin{multline} \label{eqn:ideal-world}
\mathscr{F}^{\regM\regR \rightarrow \regM\regR\regF}: \rho_{\regM\regR} \rightarrow (\mathds{1}_M \otimes \mathscr{U}^{acc}_R)\rho_{MR} \otimes \accstate
+ tr_{\regM}((\mathds{1}_M \otimes \mathscr{U}^{rej}_R)\rho_{\regM\regR}) \Omega_{M}\otimes \rejstate.
\end{multline}

\begin{defn}[Security of quantum message authentication] \label{defn:QAS-security}
Let \\ \mbox{$\{(\E_k^{\regM \rightarrow \regC},\D_k^{\regC \rightarrow \regM\regF}) \mid k \in \K \}$}~be a quantum message authentication scheme, with keys~$k$ chosen from $\K$. Then the scheme is $\epsilon$-secure if for all attacks, there exists a simulator such that:\looseness=-1
\begin{equation}
 D\Big(\frac{1}{\abs{\K}}\sum_{k \in \K}\mathscr{E}_k(\rho_{\regM\regR}), \mathscr{F}(\rho_{\regM\regR})\Big ) \leq \epsilon, \forall \rho_{\regM\regR}.
\end{equation}
Furthermore, we require that if $\mathscr{E}_k$ is polynomial-time in the size of the input register $\regM$, then~$\mathscr{F}$ is also polynomial-time in the size of the input register, $\regM$.
\end{defn}

We note that this definition is similar to the definition in~\cite{DNS12}; however we require a \emph{polynomial-time simulation} whenever the attack is polynomial-time. This does not limit the proof to polynomial-time attacks, but merely restricts the simulator to have at most the complexity of the attack. This condition being satisfied is typically a crucial ingredient in order for the composability to carry through~\cite{Unr10}.

\section{Quantum Message Authentication Schemes}
\label{sec:schemes}

Here, we present two quantum message authentication schemes, the \emph{Clifford} code (\cref{sec:CliffordCode}) and the \emph{trap} code (\cref{sec:TrapCode}). The two encoding procedures both proceed by appending trap qubits (in a fixed state) to the message register, and then \emph{twirling} by a Clifford (for the Clifford code) or a Pauli (for the trap code). The trap code also has a permutation in addition to the Pauli twirl acting on the message register. Decoding simply consists of undoing the permutation in the trap code and then in both cases measuring the traps to check for any sign of tampering. In the case of the Clifford code, only one set of traps (all in the same state) is needed because the Clifford twirl breaks any Pauli attack into a uniform mixture of Paulis which is detected on the traps with high probability. The trap code, however, relies on two sets of traps (in two different states) with both a Pauli twirl and a permutation of the message and trap qubits. Furthermore, the trap scheme requires that we first \emph{encode} the input message into an error correcting code (essentially, this is because the Pauli twirl is not as powerful as the Clifford twirl and will catch only high-weight Pauli attacks with the error correcting code taking care of the low-weight ones).

\subsection{The Clifford Code}\label{sec:CliffordCode}

We define a message authentication scheme using a Clifford encryption as follows:
\begin{enumerate}
\item The encoding, $\mathcal{E}_k^{M \rightarrow C}$, takes as input an $n$-qubit message in the $M$ system; it appends an additional $d$-qubit trap register in the state~$\ket{0}\bra{0}^{\otimes d}$. A uniformly random Clifford is then applied to the resulting $n+d$-qubit register, according to the key, $k$. The output register is called~$C$.

Mathematically, the encoding, $\mathcal{E}_{k}^{\regM \rightarrow \regC}$, indexed by a secret key, $k$, on input $\rho_{M}$ (where $C_k$ the $k^{\text{th}}$ Clifford) is given by:
\begin{equation}
\mathcal{E}_{k}: \rho_{M} \mapsto C_{k} (\rho_{M} \otimes \ket{0}\bra{0}^{\otimes d})C_{k}^{\dagger}.
\end{equation}
\item The decoding, $\mathcal{D}_k^{C \rightarrow MF}$, takes the $C$ register and applies the inverse Clifford, according to the key, $k$. The last $d$ qubits are then measured in the computational basis. If this measurement returns $\ket{0}\bra{0}^{\otimes d}$ then an additional qubit $\accstate$ is appended in the flag system, $F$. If the measurements return anything else, then the remaining system, $M$, is traced out and replaced with a fixed $n$-qubit state, $\Omega_M$, and an additional qubit, $\rejstate$, is appended in the flag system.

Mathematically, the decoding, $\mathcal{D}_{k}^{\regC \rightarrow \regM\regF}$, also indexed by the secret key, $k$, is given by:
\begin{multline}
\mathcal{D}_{k} : \rho_{C}  \mapsto tr_{0}(\mathcal{P}_{acc}C_{k}^{\dagger}(\rho_{C})C_{k}\mathcal{P}_{acc}^{\dagger})\otimes \accstate
 + tr_{M,0}(\mathcal{P}_{rej} C_{k}^{\dagger}(\rho_{C})C_{k} \mathcal{P}_{rej}^{\dagger})\Omega_M \otimes \rejstate,
\end{multline}
where $\mathcal{P}_{acc}= \mathds{1}^{\otimes n} \otimes \ket{0}\bra{0}^{\otimes d}$ and $\mathcal{P}_{rej}= \mathds{1}^{\otimes (n+d)}-\mathcal{P}_{acc}$ are measurement projectors representing the trap qubits being in their initial states or altered, respectively. Finally, $tr_{0}$ refers to the trace over the $d$ trap qubits.
\end{enumerate}

\subsection{The Trap Code}\label{sec:TrapCode}
We  define a trap code message authentication scheme as follows:

\begin{enumerate}
\item The encoding, $\mathcal{E}_k^{M \rightarrow C}$, takes as  input $\rho_{M}$ and applies an $[[n,1,d]]$-error correcting code to the single-qubit $M$ register, which will correct up to $t$ errors (where $d=2t+1$). It then appends two additional $n$-qubit trap registers, the first in the state $\ket{0}\bra{0}^{\otimes n}$ and the second in the state $\ket{+}\bra{+}^{\otimes n}$. The resulting $3n$-qubit register is then permuted and a Pauli encryption is applied, according to the key, $k$. The resulting register is called~$C$.

Mathematically the encoding, $\mathcal{E}_{k}^{\regM \rightarrow \regC}$, indexed by a two-part secret key $k=(k_1, k_2)$ is given by:
\begin{equation}
\mathcal{E}_{k}: \rho_{M} \mapsto P_{k_2}\pi_{k_1}(Enc_M(\rho_{M}) \otimes \ket{0}\bra{0}^{\otimes n} \otimes \ket{+}\bra{+}^{\otimes n})\pi_{k_1}^{\dagger}P_{k_2},
\end{equation}
where $Enc_M(\rho_{M})$ represents the input state after the error correcting code has been applied to the $M$ system, $\pi_{k_1}$ is the $k_1^{th}$ permutation and $P_{k_2}$ is the $k_2^{th}$ Pauli matrix.

We note that we use the error-correcting properties of the code only (it is sufficient in our context to simply correct low-weight Paulis on the message, as opposed detecting them and rejecting).

\item The decoding, $\mathcal{D}_k^{C \rightarrow MF}$, takes the $C$ register and applies the inverse Pauli and then the inverse permutation according to the key, $k$. The last $n$ qubits are then measured in the Hadamard basis and the second last $n$ qubits are measured in the computational basis. If these two measurements return $\ket{+}\bra{+}^{\otimes n}$ and $\ket{0}\bra{0}^{\otimes n}$ respectively, then an additional qubit $\accstate$ is appended in the flag system $F$ and the resulting $M$ register is decoded (according to the error correcting code applied in the encoding). If the measurements return anything else, then the remaining system $M$ is traced out and replaced with a fixed single-qubit state $\Omega_M$ and an additional qubit, $\rejstate$, is appended in the flag system.

Define $\mathds{P}_{\mathscr{E}}= \{ P \otimes R \otimes Q | P \in \mathds{P}_{n}, R \in \{I,Z\}^{\otimes n}, Q \in \{I,X\}^{\otimes n} \}$. Then define the measurement projector corresponding to the protocol accepting as $\mathcal{P}_{acc}= \mathds{1}^{\otimes n} \otimes \ket{0}\bra{0}^{\otimes n} \otimes \ket{+}\bra{+}^{\otimes n}$. The accepted states are then the states that can be achieved by applying any $P \in \mathds{P}_{\mathscr{E}}$ to $\rho_M \otimes \ket{0}\bra{0}^{\otimes n} \otimes \ket{+}\bra{+}^{\otimes n}$. We define $\mathcal{P}_{rej}= \mathds{1}^{\otimes 3n}-\mathcal{P}_{acc}$, the measurement projector corresponding to the protocol rejecting, where the states achieved by applying any $P \in \mathds{P}_{3n} \setminus \mathds{P}_{\mathscr{E}}$ to $Enc_M(\rho_M) \otimes \ket{0}\bra{0}^{\otimes n} \otimes \ket{+}\bra{+}^{\otimes n}$ are rejected.

Mathematically, the decoding, $\mathcal{D}_{k}^{\regC \rightarrow \regM\regF}$, also indexed by the two-part secret key, $k$, is given by:
\begin{multline}
\mathcal{D}_{k} : \rho_{C}  \mapsto Dec_M tr_{0,+}(\mathcal{P}_{acc}\pi^{\dagger}_{k_1}P_{k_2}(\rho_{C})P_{k_2}\pi_{k_1}\mathcal{P}_{acc}^{\dagger})\otimes \accstate \\
+ tr_{M,0,+}(\mathcal{P}_{rej} \pi^{\dagger}_{k_1}P_{k_2}(\rho_{C})P_{k_2} \pi_{k_1}\mathcal{P}_{acc}^{\dagger})\Omega_M \otimes \rejstate,
\end{multline}
where $Dec_M$ is the decoding of the error correcting code applied in the encryption and $tr_{0,+}$ refers to the trace over the last two sets of $n$ trap qubits.
\end{enumerate}

\section{Security of Quantum Message Authentication Schemes}
\label{sec:security}

In this section, we present simulation-based proofs for the Clifford (\cref{sec:Clifford-proof}) and the trap (\cref{sec:trap-proof}) codes. At a high level, the security of the two codes is analyzed in very similar ways (see the discussion in \cref{sec:intro}).  The main idea (in both cases) is to use a simulator that replaces the encoded message in $\regC$ with half EPR pairs, without encryption in the Clifford code, and with only a permutation in the trap code; the attack is then applied to these half EPR pairs, as well as any reference system~$\regR$. From there we are able to compare the accepted and rejected states between the real world and ideal protocols in order to find the upper bound for the trace distance between them. We will notice that these differences are the cases where the real world protocol accepts something that the simulator rejects. Specifically, this is where an attack gets through and changes a logical qubit but is not detected in the traps. Of course, these same states are not rejected by the real world protocol but they are rejected by the simulator. Because the Clifford twirl maps any non-identity Pauli attack to a uniform mixture of non-identity Paulis, the bound for this distance is simple to compute in the case of the Clifford code. In the case of the trap code, a more complicated argument is needed based on permuting the attack and a combinatorial argument that bounds the undetected attacks that can alter the logical data.

\subsection{Security of the Clifford Code}
\label{sec:Clifford-proof}

\subsubsection{Simulator.}

Recall (\cref{sec:defns}) that the simulator interacts with the ideal functionality by only altering the reference system and selecting either \emph{accept} or \emph{reject}. Given the attack, $U_{CR}$, to which the simulator has access, the simulator will apply the attack to half EPR pairs in place of the $C$ system and then perform a Bell basis measurement on the EPR pairs. It will select \emph{accept} if the EPR pairs are still in their original state, and \emph{reject} otherwise.  Let  $\mathcal{P}_{acc}^{\mathscr{U}}=\mathds{1}_{MR} \otimes \ket{\Phi^{+}}\bra{\Phi^{+}}^{\otimes (n+d)}_{C_1C_2}$ and $\mathcal{P}_{rej}^{\mathscr{U}}=\mathds{1} - \mathcal{P}_{acc}^{\mathscr{U}}$. The ideal channel is then:
\begin{multline}
\mathscr{F}^{\regM\regR \rightarrow \regM\regR\regF}: \rho_{\regM\regR} \rightarrow \\
 tr_{C_1C_2} (\mathcal{P}_{acc}^{\mathscr{U}}U_{C_1R}(\rho_{MR} \otimes \ket{\Phi^{+}}\bra{\Phi^{+}}^{\otimes (n+d)}_{C_1C_2})U_{C_1R}^{\dagger}\mathcal{P}_{acc}^{\mathscr{U}\dagger}) \otimes \accstate \\
 + tr_{\regM}(tr_{C_1C_2} (\mathcal{P}_{rej}^{\mathscr{U}}U_{C_1R}(\rho_{MR} \otimes \ket{\Phi^{+}}\bra{\Phi^{+}}^{\otimes (n+d)}_{C_1C_2})  
  U_{C_1R}^{\dagger}\mathcal{P}_{rej}^{\mathscr{U}\dagger})) \Omega_{M}\otimes \rejstate.
\end{multline}

According to the above, we define $\mathscr{U}^{acc}$ and $\mathscr{U}^{rej}$ that satisfy \cref{eqn:ideal-world} as:
\begin{equation}
\mathscr{U}^{acc}: \rho_{R}  \rightarrow tr_{C_1C_2} (\mathcal{P}_{acc}^{\mathscr{U}}U_{C_1R}(\rho_{R} \otimes \ket{\Phi^{+}}\bra{\Phi^{+}}^{\otimes (n+d)}_{C_1C_2})U_{C_1R}^{\dagger}\mathcal{P}_{acc}^{\mathscr{U}\dagger}),
\end{equation}
and
\begin{equation}
\mathscr{U}^{rej}: \rho_{R}  \rightarrow tr_{C_1C_2} (\mathcal{P}_{rej}^{\mathscr{U}}U_{C_1R}(\rho_{R} \otimes \ket{\Phi^{+}}\bra{\Phi^{+}}^{\otimes (n+d)}_{C_1C_2})U_{C_1R}^{\dagger}\mathcal{P}_{rej}^{\mathscr{U}\dagger}).
\end{equation}

For a fixed attack $U_{CR} = \sum\limits_{P \in \mathds{P}_{n+d}}\alpha_P P_C \otimes U_R^P$, with $\sum\limits_{P \in \mathds{P}_{n+d}} \abs{\alpha_P}^2=1$, we note the effects of $\mathscr{U}^{acc}$ and $\mathscr{U}^{rej}$, recalling, of course, that $\mathscr{U}^{acc}(\rho_{\regM\regR})$ is understood to be $(\mathds{1}_{M}\otimes \mathscr{U}^{acc})(\rho_{\regM\regR})$, with the same understanding for $\mathscr{U}^{rej}$:
 \begin{align}
 \mathscr{U}^{acc}(\rho_{\regM\regR})  & = tr_{C_1C_2} (\mathcal{P}_{acc}^{\mathscr{U}}U_{C_1R} (\rho_{MR} \otimes \ket{\Phi^{+}}\bra{\Phi^{+}}^{\otimes (n+d)}_{C_1C_2})U_{C_1R}^{\dagger}\mathcal{P}_{acc}^{\mathscr{U}\dagger}) \notag \\
& =  \abs{\alpha_{\mathds{1}}}^2 (\mathds{1}_M \otimes U_R^{\mathds{1}}) \rho_{MR} (\mathds{1}_M \otimes U_R^{\mathds{1}\dagger})   \label{eqn:Cliff-Uacc} \\
 \mathscr{U}^{rej}(\rho_{\regM\regR}) & = tr_{C_1C_2} (\mathcal{P}_{rej}^{\mathscr{U}} \Big(\sum\limits_{P\neq \mathds{1}}\abs{\alpha_P}^2 P_{C_1} \otimes U_R^P \Big)  \notag  \\
& \qquad (\rho_{MR} \otimes \ket{\Phi^{+}}\bra{\Phi^{+}}^{\otimes (n+d)}_{C_1C_2})  \Big(\sum\limits_{P\neq \mathds{1}}\abs{\alpha_P}^2 P_{C_1} \otimes U_R^{P\dagger} \Big)    \mathcal{P}_{rej}^{\mathscr{U}\dagger}) \notag \\
& = \sum\limits_{P\neq \mathds{1}}\abs{\alpha_P}^2 (\mathds{1}_M \otimes U_R^P)(\rho_{MR})( \mathds{1}_M \otimes U_R^{P\dagger}).\label{eqn:Cliff-Urej}
  \end{align}

We are now ready to state and prove our main theorem on the security of the Clifford message authentication scheme.
\begin{theorem}
Let $\{(\E_k^{\regS \rightarrow \regC},\D_k^{\regC \rightarrow \regS\regF}) \mid k \in \K \}$ be the Clifford quantum message authentication scheme, with parameter~$d$. Then the Clifford code is an $\epsilon$-secure quantum authentication scheme, for $\epsilon \leq \frac{3}{2^d}$.
\end{theorem}

\begin{proof}
We will follow the proof structure used in \cite{DNS12,ABE10}.

Using the simulator described above, we wish to show that:
\begin{equation}
\label{eqn:Cliff-trace-distance}
D\Big( \frac{1}{\abs{\K}} \sum\limits_{k \in \K} \mathscr{E}_k(\rho_{MR}), \mathscr{F}(\rho_{MR}) \Big) \leq \epsilon, \forall \rho_{MR}.
\end{equation}

Consider a general attack $U_{CR}$, written as  $U_{CR} = \sum\limits_{P \in \mathds{P}_{n+d}}\alpha_P P_C \otimes U_R^P$ where $\sum\limits_{P \in \mathds{P}_{n+d}} \abs{\alpha_P}^2=1$. The real-world channel is then represented as:
\begin{multline}\label{eqn:Cliff-real-w-attack}
{\mathscr{E}_k}^{\regM\regR \rightarrow \regM\regR\regF}: \rho_{\regM\regR} \mapsto \mathcal{D}_{k}\Big(\Big(\sum\limits_{P \in \mathds{P}_{n+d}}\alpha_P P_C \otimes U_R^P\Big) \mathcal{E}_{k}(\rho_{\regM\regR}) 
 \Big(\sum\limits_{P \in \mathds{P}_{n+d}}\overline{\alpha_P} P_C \otimes U_R^{P\dagger}\Big)\Big).
\end{multline}

We will use $\psi=\rho_{MR} \otimes \ket{0}\bra{0}^{\otimes d}$ to simplify the following expressions. Consider the effect of the real protocol on input $\rho_{MR}$ with attack $\sum\limits_{P \in \mathds{P}_{n+d}}\alpha_P P_C \otimes U_R^P$, conditioned on acceptance:
\begin{multline}
\frac{1}{\abs{\mathcal{K}}}\sum\limits_{k \epsilon \mathcal{K}} tr_{0}\Big(\mathcal{P}_{acc}C_{k}^{\dagger}\Big(\sum\limits_{P \in \mathds{P}_{n+d}}\alpha_P P_C \otimes U_R^P\Big)(C_{k}\psi C_{k}^{\dagger}) 
 \Big(\sum\limits_{P \in \mathds{P}_{n+d}}\overline{\alpha_P}P_C^{\dagger} \otimes U_R^{P\dagger}\Big)C_{k} \mathcal{P}_{acc}^{\dagger}\Big) \otimes \accstate.
\end{multline}
Now we can apply the Clifford Twirl (\cref{lem:Cliff-twirl}), since the sum over all keys is, of course, the sum over all Cliffords (since the keys index all $n+d$-qubit Cliffords) and then simply split the sum over all Paulis into the case with the identity Pauli from the attack, and all other Paulis. What we are left with is:
\begin{align}
& \frac{1}{\abs{\mathcal{K}}}\sum\limits_{k \epsilon \mathcal{K}} tr_{0}\Big(\sum\limits_{P \in \mathds{P}_{n+d}}\abs{\alpha_P}^2 \mathcal{P}_{acc}C_{k}^{\dagger}(P_C \otimes U_R^P)(C_{k}\psi C_{k}^{\dagger}) 
 (P_C^{\dagger} \otimes U_R^{P\dagger})C_{k} \mathcal{P}_{acc}^{\dagger}\Big) \otimes \accstate \notag \\
& =\frac{1}{\abs{\mathcal{K}}}\sum\limits_{k \epsilon \mathcal{K}} tr_{0}\Big(\abs{\alpha_{\mathds{1}}}^2 \mathcal{P}_{acc}C_{k}^{\dagger}(\mathds{1}_C \otimes U_R^{\mathds{1}})(C_{k}\psi C_{k}^{\dagger})  
 (\mathds{1}_C \otimes U_R^{\mathds{1} \dagger})C_{k} \mathcal{P}_{acc}^{\dagger}\Big) \otimes \accstate \notag \\
& \quad +\frac{1}{\abs{\mathcal{K}}}\sum\limits_{k \epsilon \mathcal{K}} tr_{0}\Big(\sum\limits_{P\neq\mathds{1}}\abs{\alpha_P}^2\mathcal{P}_{acc}C_{k}^{\dagger}(P_C \otimes U_R^P)(C_{k}\psi C_{k}^{\dagger}) 
 (P_C^{\dagger} \otimes U_R^{P\dagger})C_{k} \mathcal{P}_{acc}^{\dagger}\Big) \otimes \accstate.
\end{align}
Clearly the first term is exactly what the simulator will accept, and the second term is in exactly the right form to use a Clifford Randomization (\cref{lem:Cliff-decomp}), resulting in:
\begin{multline}
=\mathscr{U}^{acc}(\rho_{MR}) \otimes 
+ \frac{1}{\abs{\mathcal{C}_n}} tr_{0}\Big(\sum\limits_{\tilde P \neq \mathds{1}}\sum\limits_{P\neq\mathds{1}} \abs{\alpha_P}^2 \frac{\abs{\mathcal{C}_n}}{\abs{\mathds{P}_n}-1} \mathcal{P}_{acc}(\tilde P_C \otimes U_R^P)\psi  
 (\tilde P_C^{\dagger} \otimes U_R^{P\dagger})\mathcal{P}_{acc}^{\dagger}\Big)\otimes \accstate.
\end{multline}
The $\tilde P$s are the results of the Clifford Randomization applied to a Pauli, $P$. The randomization is not applied to the reference system, so the $U_R^P$ terms are not changed by the randomization.
We can use the properties of the trace to move the trace inside the first sum, and we can move the $\frac{\abs{\mathcal{C}_n}}{\abs{\mathds{P}_n}-1}$ coefficient out of both of the sums:
\begin{equation}
=\mathscr{U}^{acc}(\rho_{MR}) \otimes \accstate  
+ \frac{1}{\abs{\mathcal{C}_n}} \frac{\abs{\mathcal{C}_n}}{\abs{\mathds{P}_n}-1} \Big( \sum\limits_{\tilde P \neq \mathds{1}} tr_{0} \sum\limits_{P\neq\mathds{1}} \abs{\alpha_P}^2 \mathcal{P}_{acc}(\tilde P_C \otimes U_R^P)\psi 
 (\tilde P_C^{\dagger} \otimes U_R^{P\dagger})\mathcal{P}_{acc}^{\dagger}\Big)\otimes \accstate.
\end{equation}
We recognize the $R$ register in the second sum as the states that the simulator will reject. Recall that the simulator is in terms of the sum over all non-identity Paulis and includes the $\alpha_P$ coefficients. We can therefore write the previous line in terms of the simulator as:
\begin{equation}
=\mathscr{U}^{acc}(\rho_{MR}) \otimes \accstate 
+ \frac{1}{\abs{\mathds{P}_{n+d}}-1}\Big(\sum\limits_{\tilde P\neq\mathds{1}} tr_{0} \mathcal{P}_{acc}(\tilde P_C (\mathscr{U}^{rej}(\rho_{MR})  
\otimes \ket{0}\bra{0}^{\otimes d}) \tilde P_C^{\dagger}) \mathcal{P}_{acc}^{\dagger}\Big) \otimes \accstate.
\end{equation}
If we let $\mathds{P}_t$ be the set of all Paulis that do not alter the trap qubits, then when we apply $\mathcal{P}_{acc}$ to the above, we end up with the sum over the $\tilde P \in \mathds{P}_t \setminus \{ \mathds{1} \}$. Therefore the previous line can be simplified to:
\begin{equation}
=\mathscr{U}^{acc}(\rho_{MR}) \otimes \accstate 
 + \frac{1}{\abs{\mathds{P}_{n+d}}-1}\sum\limits_{\tilde P \in \mathds{P}_t \setminus \{\mathds{1}\}}tr_{0}(\tilde P_C (\mathscr{U}^{rej}(\rho_{MR}) \otimes \ket{0}\bra{0}^{\otimes d})\tilde P_C^{\dagger})\otimes \accstate.
\end{equation}

The effect of the real protocol on input $\rho_{MR}$ with attack $\sum\limits_{P \in \mathds{P}_{n+d}}\alpha_P P_C \otimes U_R^P$, conditioned on rejection, can be manipulated in the same way:
\begin{align}
&\frac{1}{\abs{\mathcal{K}}}\sum\limits_{k \epsilon \mathcal{K}}\Big( tr_{M,0}\Big(\mathcal{P}_{rej}\C_{k}^{\dagger}\Big(\sum\limits_{P\in \mathds{P}_{n+d}}\alpha_P P_C \otimes U_R^P\Big)(C_{k}(\psi)C_{k}^{\dagger})  
\Big(\sum\limits_{P \in \mathds{P}_{n+d}} \overline{\alpha_P} P_C^{\dagger} \otimes U_R^{P\dagger}\Big)C_{k} \mathcal{P}_{rej}^{\dagger}\Big)\Big)\Omega_M \otimes \rejstate
\notag \\
&=\frac{1}{\abs{\mathcal{K}}}\sum\limits_{k \epsilon \mathcal{K}} \Big(tr_{M,0}(\abs{\alpha_{\mathds{1}}}^2 \mathcal{P}_{rej}C_{k}^{\dagger}(\mathds{1}_C \otimes U_R^{\mathds{1}})(C_{k}(\psi)C_{k}^{\dagger}) 
(\mathds{1}_C \otimes U_R^{\mathds{1} \dagger})C_{k} \mathcal{P}_{rej}^{\dagger})\Big)\Omega_M \otimes \rejstate \notag \\
&\quad + \frac{1}{\abs{\mathcal{K}}}\sum\limits_{k \epsilon \mathcal{K}} \Big(tr_{M,0}\Big(\sum\limits_{P\neq\mathds{1}}\abs{\alpha_P}^2 \mathcal{P}_{rej}C_{k}^{\dagger}(P_C \otimes U_R^P)(C_{k}(\psi)C_{k}^{\dagger}) 
(P_C^{\dagger} \otimes U_R^{P\dagger})C_{k} \mathcal{P}_{rej}^{\dagger}\Big)\Big)\Omega_M \otimes \rejstate \notag \\
&= \frac{1}{\abs{\mathds{P}_{n+d}}-1} \sum\limits_{\tilde P \neq \mathds{1}} \sum\limits_{P \neq \mathds{1}} \abs{\alpha}^2 \Big(tr_{M,0}(\mathcal{P}_{acc}( \tilde{P}_C \otimes U_R^P)(\psi) 
  (\tilde{P}_C^{\dagger} \otimes U_R^{P\dagger})\mathcal{P}_{acc}^{\dagger})\Big)\Omega_M \otimes \rejstate \notag \\
&=tr_M(\mathscr{U}^{rej}(\rho_{MR}))\Omega_M \otimes \rejstate 
-\frac{1}{\abs{\mathds{P}_{n+d}}-1} tr_M \Big(\sum\limits_{P \in \mathds{P}_t \setminus \{ \mathds{1}\}} \mathscr{U}^{rej}(\rho_{MR})\Big) \Omega_M \rejstate \notag \\
&= tr_M(\mathscr{U}^{rej}(\rho_{MR}))\Omega_M \otimes \rejstate 
 -\frac{4^n2^d-1}{\abs{\mathds{P}_{n+d}}-1} tr_M(\mathscr{U}^{rej}(\rho_{MR}))\Omega_M \otimes \rejstate.
\end{align}
When we combine the accepted states and the rejected states into the real world protocol given by \cref{eqn:Cliff-real-w-attack}, we can write it in terms of the simulator as:
\begin{align}
\mathcal{D}&_k(U_{CR} \mathcal{E}_k (\rho_{MR}) U_{CR}^{\dagger}) \notag \\
=&\mathscr{U}^{acc}(\rho_{MR}) \otimes \accstate 
+ \frac{1}{\abs{\mathds{P}_{n+d}}-1}\sum\limits_{\tilde P \in \mathds{P}_t \setminus \{ \mathds{1} \}}tr_{0}(\tilde P_C (\mathscr{U}^{rej}(\rho_{MR}) \otimes \ket{0}\bra{0}^{\otimes d})\tilde P_C^{\dagger})\otimes \accstate \notag \\
 &+ tr_M(\mathscr{U}^{rej}(\rho_{MR}))\Omega_M \otimes \rejstate 
  -\frac{4^n2^d-1}{\abs{\mathds{P}_{n+d}}-1} tr_M(\mathscr{U}^{rej}(\rho_{MR}))\Omega_M \otimes \rejstate.
\end{align}

We can therefore write \cref{eqn:Cliff-trace-distance} as:
\begin{multline*}
\frac{1}{2}\Big \lVert \mathscr{U}^{acc}(\rho_{MR}) \otimes \accstate 
+ \frac{1}{\abs{\mathds{P}_{n+d}}-1}\sum\limits_{\tilde P \in \mathds{P}_t \setminus \{ \mathds{1} \}}tr_{0}(\tilde P_C (\mathscr{U}^{rej}(\rho_{MR}) \otimes \ket{0}\bra{0}^{\otimes d})\tilde P_C^{\dagger})\otimes \accstate  \\
 + tr_M(\mathscr{U}^{rej}(\rho_{MR}))\Omega_M \otimes \rejstate 
 -\frac{4^n2^d-1}{\abs{\mathds{P}_{n+d}}-1} tr_M(\mathscr{U}^{rej}(\rho_{MR}))\Omega_M \otimes \rejstate  \\
 - (\mathscr{U}^{acc}(\rho_{MR}) \otimes \accstate + tr_M(\mathscr{U}^{rej}(\rho_{MR}))\Omega_M \otimes \rejstate) \Big \rVert_{1}  \\
\end{multline*}
\begin{multline}
=\frac{1}{2}\Big \lVert  \frac{1}{\abs{\mathds{P}_{n+d}}-1}\sum\limits_{\tilde P \in \mathds{P}_t \setminus \{ \mathds{1}\} }tr_{0}(\tilde P_C (\mathscr{U}^{rej}(\rho_{MR}) \otimes \ket{0}\bra{0}^{\otimes d})\tilde P_C^{\dagger})\otimes \accstate  \\
- \frac{4^n2^d-1}{\abs{\mathds{P}_{n+d}}-1} tr_M(\mathscr{U}^{rej}(\rho_{MR}))\Omega_M \otimes \rejstate   \Big \rVert_{1}
\end{multline}
Since $\abs{\mathds{P}_t \setminus \{ \mathds{1} \}}=4^n2^d-1$, and the maximum trace distance between two states is $1$, we can see that by the triangle inequality, the above is bounded~by:\looseness=-1
\begin{align}
&\leq \frac{4^n2^d-1}{\abs{\mathds{P}_{n+d}}-1}  \notag \\
&=  \frac{4^n2^d-1}{4^{n+d}-1} =  \frac{1 - \frac{1}{4^n2^d}}{2^d - \frac{1}{4^n2^d}} \notag \\
&\leq 3 \times \frac{1}{2^d}.
\end{align}
This concludes the proof, showing that the Clifford code is $\frac{3}{2^d}$-secure.
\end{proof}
This is identical to the bound of $\frac{6}{2^d}$ achieved in \cite{DNS12} when we consider that we use the trace distance in our definition of security, and \cite{DNS12} uses the trace norm, which differs from the trace distance by a factor of $2$.

\subsection{Security of the Trap Code}
\label{sec:trap-proof}

\subsubsection{Simulator.}
Recall (\cref{sec:defns}) that the simulator interacts with the ideal functionality by only altering the reference system and selecting either \emph{accept} or \emph{reject}. Given the attack, $U_{CR}$, to which the simulator has access, the simulator will apply the attack to randomly permuted half EPR pairs in place of the $C$ system and then de-permute the EPR pairs and perform a Bell basis measurement. It will select \emph{accept} if the first $n$ of the EPR pairs have $\leq t$ errors, the next $n$ of the EPR pairs are either unchanged or have phase flip errors, and the last $n$ of the EPR pairs are either unchanged or have bit flip errors. It will select \emph{reject} otherwise. Let $\mathds{P}_{\mathscr{F}}=\{ P \otimes R \otimes Q |  P \in \mathds{P}_{n}, \omega(P) \leq t, R \in \{I,Z\}^{\otimes n}, Q \in \{I,X\}^{\otimes n} \} $. Specifically, $\mathds{P}_{\mathscr{F}}$ is the set of all Paulis that the ideal protocol will accept being applied to the half EPR pair---Paulis that would apply at most $t$ non-identity Paulis on the message space and would not alter the $\ket{0}\bra{0}^{\otimes n}$ or the $\ket{+}\bra{+}^{\otimes n}$ traps in the real world protocol. Finally, define the measurement projector corresponding to the simulator selecting \emph{accept} as:
\begin{align}
\mathcal{P}_{acc}^{\mathscr{U}} &=\sum\limits_{Q \in \{I,X\}^{\otimes n}} \sum\limits_{R \in \{I,Z\}^{\otimes n}}\sum\limits_{P \in \mathds{P}_{n} \mid \omega(P) \leq t} \mathds{1}_{MR} \otimes (P \otimes R \otimes Q)_{C_1} \notag \\
& \qquad \qquad\qquad\qquad\qquad\qquad  \ket{\Phi^{+}}\bra{\Phi^{+}}_{C_1C_2}^{\otimes 3n} (P \otimes R \otimes Q)_{C_1}  \notag \\
&=\sum_{P \in \mathds{P}_{\mathscr{F}}}  \mathds{1}_{MR} \otimes (P_{C_1} \ket{\Phi^{+}}\bra{\Phi^{+}}_{C_1C_2}^{\otimes 3n}P_{C_1}^\dag),
\end{align}
and the measurement projector corresponding to the simulator selecting \emph{reject} as:
\begin{equation}
\mathcal{P}_{rej}^{\mathscr{U}}=\mathds{1} - \mathcal{P}_{acc}^{\mathscr{U}}.
\end{equation}

The ideal channel with attack $U_{C_1\regR}$ is therefore:
\begin{align}
&\mathscr{F}^{\regM\regR \rightarrow \regM\regR\regF}:\\
&\rho_{\regM\regR} \rightarrow tr_{C_1C_2} \frac{1}{\abs{\Pi_{3n}}}\sum_{\pi \in \Pi_{3n}} \Big(\mathcal{P}_{acc}^{\mathscr{U}} \pi^\dagger_{C_1} U_{C_1R} \pi_{C_1} 
 (\rho_{MR} \otimes \ket{\Phi^{+}}\bra{\Phi^{+}}^{\otimes 3n}_{C_1C_2}) \pi_{C_1}^{\dagger}  U_{C_1R}^{\dagger} \pi_{C_1} \mathcal{P}_{acc}^{\mathscr{U}\dagger}\Big) \otimes \accstate \notag \\
& + tr_{\regM}\Big(tr_{C_1C_2} \frac{1}{\abs{\Pi_{3n}}}  \sum_{\pi \in \Pi_{3n}}  \Big(\mathcal{P}_{rej}^{\mathscr{U}} \pi_{C_1}^{\dagger} U_{C_1R} \pi_{C_1} (\rho_{MR} \otimes \ket{\Phi^{+}}\bra{\Phi^{+}}^{\otimes 3n}_{C_1C_2}) 
  \pi^\dagger_{C_1}  U_{C_1R}^{\dagger} \pi_{C_1} \mathcal{P}_{rej}^{\mathscr{U}\dagger}\Big)\Big) \Omega_{M}\otimes \rejstate. \notag
\end{align}

For a fixed attack $U_{CR} = \sum\limits_{P \in \mathds{P}_{3n}}\alpha_P P_C \otimes U_R^P$, with $\sum\limits_{P \in \mathds{P}_{3n}} \abs{\alpha_P}^2 =1$ and where for the sake of brevity we will represent  $\rho_{MR} \otimes \ket{\Phi^{+}}\bra{\Phi^{+}}^{\otimes 3n}_{C_1C_2}$ with $\phi_{MRC_1C_2}$, the ideal channel becomes:
\begin{align}
&\mathscr{F}^{\regM\regR \rightarrow \regM\regR\regF}:  \rho_{\regM\regR} \rightarrow \\
&tr_{C_1C_2} \frac{1}{\abs{\Pi_{3n}}}\sum_{\pi \in \Pi_{3n}} \Bigg( \mathcal{P}_{acc}^{\mathscr{U}} \pi^\dagger_{C_1} \Big( \sum\limits_{P \in \mathds{P}_{3n}} \alpha_P P_{C_1} \otimes U_R^{P} \Big) \pi_{C_1} \phi_{MRC_1C_2}  
 \pi_{C_1}^{\dagger}  \Big( \sum\limits_{P \in \mathds{P}_{3n}} \overline{\alpha_P} P_{C_1} \otimes U_R^{P \dagger} \Big)\pi_{C_1} \mathcal{P}_{acc}^{\mathscr{U}\dagger} \otimes \accstate \notag\\
&+ tr_{\regM} \Big(\mathcal{P}_{rej}^{\mathscr{U}} \pi_{C_1}^{\dagger} \Big( \sum\limits_{P \in \mathds{P}_{3n}} \alpha_P P_{C_1} \otimes U_R^{P} \Big) \pi_{C_1} \phi_{MRC_1C_2}
 \pi^\dagger_{C_1}  \Big( \sum\limits_{P \in \mathds{P}_{3n}} \overline{\alpha_P} P_{C_1} \otimes U_R^{P \dagger} \Big) \pi_{C_1} \mathcal{P}_{rej}^{\mathscr{U}\dagger}\Big) \Omega_{M}\otimes \rejstate \Bigg).\notag
\end{align}
From here we will move the permutations to act on the attack Paulis, since they're all applied to the same register,~$C_1$:
\begin{align}
=&tr_{C_1C_2}\frac{1}{\abs{\Pi_{3n}}} \\
 &\sum_{\pi \in \Pi_{3n}}\Bigg( \Big(\mathcal{P}_{acc}^{\mathscr{U}}  \Big( \sum\limits_{P \in \mathds{P}_{3n}} \alpha_P \pi^\dagger_{C_1} P_{C_1} \pi_{C_1} \otimes U_R^{P} \Big)  \phi_{MRC_1C_2}   
 \Big( \sum\limits_{P \in \mathds{P}_{3n}} \overline{\alpha_P} \pi_{C_1}^{\dagger} P_{C_1} \pi_{C_1} \otimes U_R^{P \dagger} \Big)\mathcal{P}_{acc}^{\mathscr{U}\dagger}\Big) \otimes \accstate \notag \\
&+ tr_{\regM} \Big (\mathcal{P}_{rej}^{\mathscr{U}}  \Big( \sum\limits_{P \in \mathds{P}_{3n}} \alpha_P \pi^{\dagger}_{C_1}P_{C_1} \pi_{C_1}\otimes U_R^{P} \Big)  \phi_{MRC_1C_2}  
\Big( \sum\limits_{P \in \mathds{P}_{3n}} \overline{\alpha_P} \pi^\dagger_{C_1} P_{C_1} \pi_{C_1}\otimes U_R^{P \dagger} \Big)  \mathcal{P}_{rej}^{\mathscr{U}\dagger}\Big) \Omega_{M} \otimes \rejstate \Bigg). \notag
\end{align}
Finally we apply the projectors:
\begin{align}
=&tr_{C_1C_2} \frac{1}{\abs{\Pi_{3n}}} \notag\\
&\sum_{\pi \in \Pi_{3n}}\Bigg( \Big( \sum\limits_{P | \pi^{\dagger}P \pi \in \mathds{P}_{\mathscr{F}}} \abs{\alpha_P}^2 (\pi^\dagger_{C_1} P_{C_1} \pi_{C_1} \otimes U_R^{P})(\phi_{MRC_1C_2})  
(\pi_{C_1}^{\dagger} P_{C_1} \pi_{C_1} \otimes U_R^{P \dagger}) \Big) \otimes \accstate  \notag\\
&+  tr_{\regM} \Big( \sum\limits_{P | \pi^{\dagger} P \pi \notin \mathds{P}_{\mathscr{F}}} \abs{\alpha_P}^2 (\pi^{\dagger}_{C_1}P_{C_1} \pi_{C_1}\otimes U_R^{P}) (\phi_{MRC_1C_2}) 
 (\pi^\dagger_{C_1} P_{C_1} \pi_{C_1}\otimes U_R^{P \dagger}) \Big)  \Omega_{M}\otimes \rejstate \Bigg).
\end{align}

We are now ready to present our main theorem on the security of the trap code:
\begin{theorem}
Let $\{(\E_k^{\regS \rightarrow \regC},\D_k^{\regC \rightarrow \regS\regF}) \mid k \in \K \}$ be the trap quantum message authentication scheme with parameter~$t$, the number of bit or phase flip errors that the error correcting code applied to the input message qubit can correct. Then the trap code is an $\epsilon$-secure quantum message authentication scheme, for $\epsilon \leq (\frac{1}{3})^{t+1}$.
\end{theorem}
\begin{proof}
Using the simulator described above, we wish to show that:
\begin{equation}
\label{eqn:trace-distance}
D \Big (\frac{1}{\abs{\K}} \sum\limits_{k \in \K} \mathscr{E}_k(\rho_{MR}),\mathscr{F}(\rho_{MR}) \Big) \leq \epsilon, \forall \rho_{MR}.
\end{equation}
Consider a general attack $U_{CR}$, written as  $U_{CR} = \sum\limits_{P \in \mathds{P}_{3n}}\alpha_P P_C \otimes U_R^P$ with $\sum\limits_{P \in \mathds{P}_{3n}} \abs{\alpha_P}^2=1$. Let $\psi = Enc_M (\rho_{MR} ) \otimes \ket{0}\bra{0}^{\otimes n} \otimes \ket{+}\bra{+}^{\otimes n}$. The real-world channel is then represented as:
\begin{align}
\label{eqn:real-w-attack}
&{\mathscr{E}_k}^{\regM\regR \rightarrow \regM\regR\regF}: \rho_{\regM\regR} \mapsto
\mathcal{D}_{k}\Big(\Big(\sum\limits_{P \in \mathds{P}_{3n}}\alpha_P P_C \otimes U_R^P\Big) \mathcal{E}_{k}(\rho_{\regM\regR}) 
\Big(\sum\limits_{P \in \mathds{P}_{3n}}\overline{\alpha_P} P_C \otimes U_R^{P\dagger}\Big)\Big) \\
& = \frac{1}{\abs{\mathcal{K}}} tr_{0,+} \sum\limits_{k \epsilon \mathcal{K}}\Bigg(Dec_M\Big(\mathcal{P}_{acc}\pi^{\dagger}_{k_1}P_{k_2}\Big(\sum\limits_{P \in \mathds{P}_{3n}}\alpha_P P_C \otimes U_R^P\Big)P_{k_2}\pi_{k_1}\psi  \notag \\
& \qquad \qquad \qquad \quad \qquad \qquad \qquad \quad \qquad \quad
\pi_{k_1}^{\dagger}P_{k_2}\Big(\sum\limits_{P \in \mathds{P}_{3n}}\overline{\alpha_P} P_C \otimes U_R^{P\dagger}\Big)P_{k_2} \pi_{k_1} \mathcal{P}_{acc}^{\dagger}\Big)
 \otimes \accstate \notag\\
&+ tr_{M}\Big(\mathcal{P}_{rej}\pi^{\dagger}_{k_1}P_{k_2}\Big(\sum\limits_{P \in \mathds{P}_{3n}}\alpha_P P_C \otimes U_R^P\Big)(P_{k_2}\pi_{k_1} \psi \pi_{k_1}^{\dagger}P_{k_2}) \notag 
\Big(\sum\limits_{P \in \mathds{P}_{3n}}\overline{\alpha_P} P_C \otimes U_R^{P\dagger}\Big)P_{k_2}\pi_{k_1} \mathcal{P}_{rej}^{\dagger}\Big) \Omega_M \otimes \rejstate \Bigg).
\end{align}
From here we apply the Pauli Twirl (\cref{lem:Pauli-twirl}):
\begin{align}
& = \frac{1}{\abs{\mathcal{K}_1}} tr_{0,+} \sum\limits_{k_1 \epsilon \mathcal{K}_1}\Bigg(Dec_M \Big(\mathcal{P}_{acc}\pi^{\dagger}_{k_1}\Big(\sum\limits_{P \in \mathds{P}_{3n}}\abs{\alpha_P}^2 (P_C \otimes U_R^P)\pi_{k_1} \psi   
\pi_{k_1}^{\dagger}(P_C \otimes U_R^{P\dagger})\Big) \pi_{k_1} \mathcal{P}_{acc}^{\dagger}\Big) \otimes \accstate \notag \\
&\quad + tr_{M}\Big(\mathcal{P}_{rej}\pi^{\dagger}_{k_1}\Big(\sum\limits_{P \in \mathds{P}_{3n}}\abs{\alpha_P}^2 (P_C \otimes U_R^P)\pi_{k_1} \psi  
\pi_{k_1}^{\dagger}(P_C \otimes U_R^{P\dagger})\Big)\pi_{k_1} \mathcal{P}_{rej}^{\dagger}\Big) \Omega_M \otimes \rejstate \Bigg).
\end{align}
Since the permutations act on the same register as the attack Paulis, we can move the permutations to be considered to be acting on the Paulis instead of the message and traps:
\begin{align}
 &= \frac{1}{\abs{\mathcal{K}_1}}tr_{0,+}\sum\limits_{k_1 \epsilon \mathcal{K}_1}\Bigg( Dec_M \Big(\mathcal{P}_{acc}\Big(\sum\limits_{P \in \mathds{P}_{3n}}\abs{\alpha_P}^2 (\pi_{k_1}^{\dagger}P_C \pi_{k_1}\otimes U_R^P) \psi  
(\pi_{k_1}^{\dagger}P_C \pi_{k_1} \otimes U_R^{P\dagger})\Big)  \mathcal{P}_{acc}^{\dagger}\Big) \otimes \accstate \notag \\
&\quad +tr_{M}\Big(\mathcal{P}_{rej}\Big(\sum\limits_{P \in \mathds{P}_{3n}}\abs{\alpha_P}^2 (\pi^{\dagger}_{k_1}P_C\pi_{k_1} \otimes U_R^P) \psi  
 (\pi_{k_1}^{\dagger}P_C \pi_{k_1}  \otimes U_R^{P\dagger})\Big)\mathcal{P}_{rej}^{\dagger}\Big) \Omega_M \otimes \rejstate \Bigg).
\end{align}
Finally we apply the projectors and notice that $\mathcal{K}_1 = \Pi_{3n}$:
\begin{align}
&= \frac{1}{\abs{\Pi_{3n}}}tr_{0,+}\sum\limits_{\pi \epsilon \Pi_{3n}}\Bigg(Dec_M\Big(\sum\limits_{P | \pi^{\dagger}P\pi \in \mathds{P}_{\mathscr{E}}}\abs{\alpha_P}^2 (\pi^{\dagger}P_C \pi \otimes U_R^P)\psi  
 (\pi^{\dagger}P_C \pi \otimes U_R^{P\dagger})\Big) \otimes \accstate  \notag \\
&\quad +  tr_{M} \Big(\sum\limits_{P | \pi^{\dagger} P \pi \in \mathds{P}_{3n}\setminus \mathds{P}_{\mathscr{E}}} \abs{\alpha_P}^2 (\pi^{\dagger}P_C\pi \otimes U_R^P)\psi  
(\pi^{\dagger}P_C \pi  \otimes U_R^{P\dagger})\Big)\Omega_M \otimes \rejstate \Bigg).
\end{align}

Then:
\begin{align}
&\frac{1}{2} \Big \lVert \frac{1}{\abs{\K}}\sum\limits_{k \in \K} \mathscr{E}_k(\rho_{MR}) - \mathscr{F}(\rho_{MR}) \Big \rVert_{1}  \\
&= \frac{1}{2} \Big\lVert \frac{1}{\abs{\Pi_{3n}}}\sum\limits_{\pi \epsilon \Pi_{3n}}\Bigg( tr_{0,+}\Big(Dec_M\Big(\sum\limits_{P | \pi^{\dagger}P\pi \in \mathds{P}_{\mathscr{E}}}\abs{\alpha_P}^2 (\pi^{\dagger}P_C \pi \otimes U_R^P)\psi  
 (\pi^{\dagger}P_C \pi \otimes U_R^{P\dagger})\Big) \otimes \accstate  \notag \\
&\quad +  tr_{M} \Big(\sum\limits_{P | \pi^{\dagger} P \pi \in \mathds{P}_{3n}\setminus \mathds{P}_{\mathscr{E}}} \abs{\alpha_P}^2 (\pi^{\dagger}P_C\pi \otimes U_R^P)\psi  
 (\pi^{\dagger}P_C \pi  \otimes U_R^{P\dagger})\Big)\Omega_M \otimes \rejstate \Big) \notag \\
&\quad - tr_{C_1C_2}  \Big( \sum\limits_{P | \pi^{\dagger}P \pi \in \mathds{P}_{\mathscr{F}}} \abs{\alpha_P}^2 (\pi^\dagger_{C_1} P_{C_1} \pi_{C_1} \otimes U_R^{P})(\phi_{MRC_1C_2})  
 (\pi_{C_1}^{\dagger} P_{C_1} \pi_{C_1} \otimes U_R^{P \dagger}) \Big) \otimes \accstate \notag \\
&\quad -  tr_{\regM C_1C_2} \Big( \sum\limits_{P | \pi^{\dagger} P \pi \notin \mathds{P}_{\mathscr{F}}} \abs{\alpha_P}^2 (\pi^{\dagger}_{C_1}P_{C_1} \pi_{C_1}\otimes U_R^{P}) (\phi_{MRC_1C_2}) 
 (\pi^\dagger_{C_1} P_{C_1} \pi_{C_1}\otimes U_R^{P \dagger}) \Big)  \Omega_{M}\otimes \rejstate \Bigg) \Big \rVert_{1}. \notag
\end{align}

We will subtract the accepted states in the ideal protocol from those accepted in the real protocol and we will subtract the rejected states in the real protocol from the rejected states in the ideal protocol. Note that $\mathds{P}_{\mathscr{E}}\setminus \mathds{P}_{\mathscr{F}}=\{P \otimes R \otimes Q | P \in \mathds{P}_{n}, \omega(P)>t, R \in \{I,Z\}^{\otimes n}, Q \in \{I,X\}^{\otimes n} \}$.
\begin{align}
&= \frac{1}{2} \Big \lVert \frac{1}{\abs{\Pi_{3n}}}\sum\limits_{\pi \epsilon \Pi_{3n}} \sum\limits_{P | \pi^{\dagger}P\pi \in \mathds{P}_{\mathscr{E}}\setminus \mathds{P}_{\mathscr{F}}} \Bigg( tr_{0,+} \Big( Dec_M  (\abs{\alpha_P}^2 (\pi^{\dagger}P_C \pi \otimes U_R^P)\psi  
(\pi^{\dagger}P_C \pi \otimes U_R^{P\dagger}))\Big)\otimes \accstate \notag \\
&\quad - tr_{\regM C_1C_2} \Big( \abs{\alpha_P}^2 (\pi^{\dagger}_{C_1}P_{C_1} \pi_{C_1}\otimes U_R^{P}) (\phi_{MRC_1C_2})  
(\pi^\dagger_{C_1} P_{C_1} \pi_{C_1}\otimes U_R^{P \dagger}) \Big)  \Omega_{M}\otimes \rejstate \Bigg) \Big \rVert_{1}.
\end{align}
Here we will use the triangle inequality to remove the sums from the trace distance:
\begin{align}
&\leq \frac{1}{2} \frac{1}{\abs{\Pi_{3n}}}\sum\limits_{\pi \in \Pi_{3n}} \sum\limits_{P | \pi^{\dagger}P\pi \in \mathds{P}_{\mathscr{E}}\setminus \mathds{P}_{\mathscr{F}}} \Big \lVert tr_{0,+} \Big( Dec_M (\abs{\alpha_P}^2 (\pi^{\dagger}P_C \pi \otimes U_R^P)\psi  
(\pi^{\dagger}P_C \pi \otimes U_R^{P\dagger}))\Big) \otimes \accstate  \notag \\
& \quad -  tr_{\regM C_1C_2} \Big(\abs{\alpha_P}^2 (\pi^{\dagger}_{C_1}P_{C_1} \pi_{C_1}\otimes U_R^{P}) (\phi_{MRC_1C_2}) 
 (\pi^\dagger_{C_1} P_{C_1} \pi_{C_1}\otimes U_R^{P \dagger}) \Big)  \Omega_{M}\otimes \rejstate \Big \rVert_{1}.
\end{align}
Since the maximum trace distance between two states is $1$ we have:
\begin{align}
& \leq \frac{1}{\abs{\Pi_{3n}}}\sum\limits_{k_1 \epsilon \mathcal{K}_1} \sum\limits_{P | \pi^{\dagger}P\pi \in \mathds{P}_{\mathscr{E}}\setminus \mathds{P}_{\mathscr{F}}} \abs{\alpha_P}^2.
\end{align}
Now if we let $\eta_P$ be the number of permutations, $\pi$ of $P$ such that $\pi^{\dagger}P\pi \in \mathds{P}_{\mathscr{E}}\setminus \mathds{P}_{\mathscr{F}}$, then the above can be written as:
\begin{align}
& = \frac{1}{\abs{\Pi_{3n}}} \sum\limits_{P \in \mathds{P}_{3n}} \eta_P \times \abs{\alpha_P}^2.
\end{align}

In \cref{app:Eta}, we give a combinatorial Lemma (\cref{lem:eta}), which gives us $\eta_P \leq {n \choose t+1}(t+1)! (3n-(t+1))!$. Thus, since $\sum\limits_{P \in \mathds{P}_{3n}}\abs{\alpha_P}^2=1$, the above expression can be bounded by:
\begin{align}
& \leq  \frac{1}{(3n)!} \times {n \choose t+1}(t+1)! (3n-(t+1))! \notag \\
& = \frac{\prod\limits_{i=1}^{n}i \prod\limits_{i=1}^{3n-t-1}i}{\prod\limits_{i=1}^{n-t-1}i  \prod\limits_{i=1}^{3n}i} =  \frac{\prod\limits_{i=n-t}^{n}i }{\prod\limits_{i=3n-t}^{3n}i} =  \prod\limits_{i=0}^{t} \frac{ n-t+i}{3n-t+i} \notag \\
& \leq  \prod\limits_{i=0}^{t} \frac{1}{3}  = \Big(\frac{1}{3}\Big)^{t+1}
\end{align}

Therefore, $D \Big ( \frac{1}{\abs{\K}} \sum\limits_{k \in \K} \mathscr{E}_k(\rho_{MR}), \mathscr{F}(\rho_{MR}) \Big ) \leq (\frac{1}{3})^{t+1} , \forall \rho_{MR}$.
\end{proof}

We note that this is very similar to the bound in \cite{BGS13} of $(\frac{2}{3})^{d/2}$: note that the trap code in \cite{BGS13} uses the error \emph{detection} property of the code. Since a code of distance~$d$ can detect up to~$d/2$ errors, this bound is consistent with our bound of  $(\frac{1}{3})^{t+1}$.

\subsubsection*{Acknowledgements}
We would like to thank Florian Speelman for feedback on a prior version of this work, as well as the anonymous reviewers for useful corrections.

\ifcameraready
  \bibliographystyle{arxiv_no_month}
\else
     \bibliographystyle{alphaarxiv}
\fi
\newcommand{\etalchar}[1]{$^{#1}$}

\appendix

\section{Proof of \cref{lem:Cliff-twirl}} \label{sec:app-Clifford}
For completeness, we provide a proof of the Clifford Twirl \cref{lem:Cliff-twirl} (see also~\cite{DCEL09}).

\begin{proof}
We will follow the proof structure of \cite{DCEL09}, but simplify for our purposes. Since $\mathds{P}_n$ is a subgroup of $\mathcal{C}_n$, then we know that the number of left cosets of $\mathds{P}_n$ in $\mathcal{C}_n$, or the index of $\mathds{P}_n$ in $\mathcal{C}_n$, is given by:
\begin{equation}
[\mathcal{C}_n  : \mathds{P}_n]=\frac{\abs{\mathcal{C}_n}}{\abs{\mathds{P}_n}}
\end{equation}

Then given a representative from each of the cosets, $\{C_1, C_2,\ldots, C_{\frac{\abs{\mathcal{C}_n}}{\abs{\mathds{P}_n}}}\}$, we can rewrite the sum over all Cliffords as a double sum of the Paulis and the coset representatives as below. We note that it does not matter which representative we choose, only that we have one from each of the cosets, and that the indices on the $C$ terms give which coset they came from:

\begin{align}
\sum\limits_{C \in \mathcal{C}_n} C^{\dagger}PC \rho C^{\dagger}P'C &= \sum\limits_{i=1}^{\frac{\abs{\mathcal{C}_n}}{\abs{\mathds{P}_n}}} \sum\limits_{R \in \mathds{P}_n} (C_iR)^{\dagger}PC_iR \rho (C_iR)^{\dagger}P'C_iR \notag \\
& =   \sum\limits_{i=1}^{\frac{\abs{\mathcal{C}_n}}{\abs{\mathds{P}_n}}}  \sum\limits_{R \in \mathds{P}_n} R^{\dagger}C_i^{\dagger} PC_iR \rho R^{\dagger}C_i^{\dagger}P'C_iR
\end{align}
Now since $C_i^{\dagger}PC_i=Q_i$ for some $Q_i \in \mathds{P}_n$, and since if $P \neq P'$ then $Q_i \neq Q'_i$, we can simplify our expression to one that only involves Paulis:
\begin{align}
& =  \sum\limits_{i=1}^{\frac{\abs{\mathcal{C}_n}}{\abs{\mathds{P}_n}}} \sum\limits_{R \in \mathds{P}_n}  R^{\dagger}Q_iR \rho R^{\dagger}Q_i'R \notag \\
&= \sum\limits_{i=1}^{\frac{\abs{\mathcal{C}_n}}{\abs{\mathds{P}_n}}} 0  \notag \\
&=0
\end{align}
by the Pauli Twirl (\cref{lem:Pauli-twirl}).
\end{proof}

\section{Combinatorial Lemma}\label{app:Eta}

\begin{lemma} \label{lem:eta}
For a fixed $P \in \mathds{P}_{3n}$, let $\eta_P$ denote the number of permutations $\pi$ of~$P$ such that $\pi^{\dagger}P\pi \in \mathds{P}_{\mathscr{E}} \setminus \mathds{P}_{\mathscr{F}}$ Then for all $P$:
\begin{equation}
\eta_P \leq {n \choose t+1}(t+1)!(3n-(t+1))!\,.
\end{equation}
\end{lemma}

An intuitive argument for the above lemma is that $\eta_P$ can be upper-bounded by fixing a Pauli~$P \in \{I,X\}^{3n}$ of weight $t+1$. We show that a Pauli with greater weight will have $\leq \eta_P$ possible allowed permutations. To find the number of possible allowed permutations, we will consider the first $n$ positions, where we require at least $t+1$ non-identity Paulis (for a total of $\binom{n}{t+1} (t+1)!$ permutations). The remaining positions are then simply permuted, since we have used all of the non-identity Paulis already, contributing a multiplicative factor of  $(3n-(t+1))!$ permutations. This is formalized below (where we also consider general attack Paulis consisting of combinations of $X$, $Y$ and $Z$).

\begin{proof}

In order to find an upper bound for $\eta_P$, we look to find the Pauli, $P$, that has the largest number of permutations, $\pi$, such that $\pi^{\dagger}P \pi \in \mathds{P}_{\mathscr{E}}\setminus \mathds{P}_{\mathscr{F}}$.

For a Pauli~$P$ with $\omega(P)=d$, we write $d=d_x+d_y+d_z+x_1+y+z_1+x_2+z_2$ for values $d_x, d_y, d_z, x_1, y, z_1, x_2, z_2$ as follows:
\begin{enumerate}
\item $d_x, d_y, d_z$ where $d_x + d_y + d_z=t+1$. These are the $t+1$ $X$, $Y$, and $Z$ Paulis that must be applied to the first $n$ qubits for the Pauli to be in $\mathds{P}_{\mathscr{E}} \setminus \mathds{P}_{\mathscr{F}}$.
\item $y$ where $y+d_y$ is the total number of $Y$ Paulis in $P$ and $y$ are the additional $Y$ Paulis applied to the first $n$ qubits. Note that $Y$ Paulis cannot be applied to either set of traps without altering them.
\item $x_1, x_2$ where $x_1+x_2 + d_x$ is the total number of $X$ Paulis in $P$ and $x_1$ are the additional $X$ Paulis applied to the first $n$ qubits and $x_2$ are the $X$ Paulis applied to the $\ket{+}\bra{+}^{\otimes n}$ traps.
\item $z_1, z_2$ where $z_1+z_2 + d_z$ is the total number of $Z$ Paulis in $P$ and $z_1$ are the additional $Z$ Paulis applied to the first $n$ qubits and $z_2$ are the $Z$ Paulis applied to the $\ket{0}\bra{0}^{\otimes n}$ traps.
\end{enumerate}

Then  the possible permutations on $P$ are found by multiplying the following terms:
\begin{enumerate}
\item $\binom{n}{d_x,d_y,d_z,n-t-1}d_x!d_y!d_z!$ Which is the number of ways to choose the required $t+1$ spots for the minimum number of Paulis applied to the first $n$ qubits, multiplied by the number of ways of permuting each of the sets of $X$, $Y$, and $Z$ Paulis. Note that this term simplifies to $\frac{n!}{(n-t-1)!}$,
\item $\binom{n-t-1}{x_1} x_1!$, the number of ways to apply $x_1$ additional $X$ Paulis to the first $n$ qubits,
\item $\binom{n-t-1-x_1}{y}y!$, the number of ways to apply $y$ additional $Y$ Paulis to the first $n$ qubits,
\item $\binom{n-t-1-x_1-y}{z_1} z_1!$, the number of ways to apply $z_1$ additional $Z$ Paulis to the first $n$ qubits,
\item $\binom{n}{x_2}x_2!$, the number of ways to apply $x_2$ $X$ Paulis to the $n$ traps that will not be changed by them,
\item $\binom{n}{z_2}z_2!$, the number of ways to apply $z_2$ $Z$ Paulis to the $n$ traps that will not be changed by them, and
\item $(3n-(d_x+d_y+d_z+x_1+y+z_1+x_2+z_2))!$ the number of ways to permute the remaining identity qubits, which simplifies to $(3n-d)!$.
\end{enumerate}

The product, once simplified, is then:

\begin{align}
\eta_P &= \frac{n!n!n!(3n-d)!}{(n-t-1-x_1-y-z_1)!(n-x_2)!(n-z_2)!} \notag \\
&= \prod\limits_{n-t-x_1-y-z_1}^{n}i\prod\limits_{n-x_2+1}^{n}i\prod\limits_{n-z_2+1}^{n}i\prod\limits_{i=1}^{3n-t-1-x_1-y-z_1-x_2-z_2}i
\end{align}

Since $t$ is fixed, in order to maximize the above expression, we need to minimize $x_1, y, z_1, x_2, z_2$. This is achieved by setting $x_1=y=z_1=x_2=z_2=0$, and therefore $d=t+1$: we thus find that $\eta_P \leq \prod\limits_{n-t}^{n}i\prod\limits_{i=1}^{3n-t-1}i={n \choose t+1}(t+1)!(3n-(t+1))!$.
\end{proof}

\end{document}